\RequirePackage{amsmath}
\documentclass[runningheads,a4paper,10pt]{llncs}

\begingroup
  \def\x{\endgroup\ExecuteOptions{dvipdfm}}%
  \ifx\pdfoutput\undefined
  \else
    \ifx\pdfoutput\relax
    \else
      \ifcase\pdfoutput
        
      \else
        \def\x{\endgroup\ExecuteOptions{pdftex}}%
	
      \fi
    \fi
  \fi
\x

\usepackage{booktabs}
\usepackage{tabularx}
\usepackage{multicol}
\usepackage{multirow}

\usepackage{graphicx}
\usepackage{amsfonts}

\usepackage{amssymb}
\usepackage{array}
\usepackage{mathtools}
\usepackage{caption}
\usepackage{subcaption}

\usepackage{url}
\usepackage[utf8]{inputenc}
\usepackage[T1]{fontenc}
\usepackage{times}

\usepackage[boxed,vlined,linesnumbered,noend]{algorithm2e}
\usepackage[margin=1.4in,footskip=0.4in]{geometry}
\newtheorem{fact}{Fact}

\SetKwInOut{Input}{Input}
\SetKwInOut{RestInput}{}
\SetKwInOut{Output}{Output}
\SetKwInOut{RestOutput}{}
\SetKwInOut{Actors}{Actors}
\SetAlFnt{\small}
\newcommand{\keywords}[1]{\par\addvspace\baselineskip
\noindent\keywordname\enspace\ignorespaces#1}



\input{./commands}
\newcommand{\anonymous}[2]{#1} 

\newcommand{\compressaffil}[2]{#1} 

\usepackage[bookmarks=true,colorlinks=true]{hyperref}    

\begin{document}

\title{\papertitle\thanks{This work was partially supported by
		          Polish National Science Center grant 2015/17/B/ST6/01897.}}
\titlerunning{\papertitle}

\author{%
	\anonymous{
		\mbox{Karol Gotfryd} \and \mbox{Marek Klonowski} \and \mbox{Dominik Paj\k{a}k}
	\compressaffil{}{}
	}
	{Of anonymous authors}
}
\authorrunning{
	\anonymous{K. Gotfryd, M.~Klonowski, D.Paj\k{a}k
	}{Anonymous authors}
}

\institute{
	\anonymous{ 
		\compressaffil{
			Department of Computer Science
		}{}\\
		Wroc{\l}aw University of Science and Technology\\
		\compressaffil{
			Wybrze{\.z}e Wyspia{\'n}skiego 27, 50-370 Wroc{\l}aw, Poland
		}{}\\
	\compressaffil{\email{%
	\{firstname.secondname\}@pwr.edu.pl}}{}
}{...}}

\maketitle

\begin{abstract}
	
We consider the following problem - a group of mobile agents perform some task on a terrain modeled as a graph. 
In a given moment of time an adversary gets an access to the graph and positions of the agents. 
Shortly before adversary's observation the mobile agents have a chance to relocate themselves in order to hide their
initial configuration. We assume that the initial configuration may possibly reveal to the adversary some information
about the task they performed. Clearly agents have to change their location in possibly short time using minimal energy. 
In our paper we introduce a definition of a \emph{well hiding} algorithm in which the starting and final configurations of
the agents have small mutual information. Then we discuss the influence of various features of the model on the running
time of the optimal well-hiding algorithm. We show that if the topology of the graph is known to the agents, then the number
of steps proportional to the diameter of the graph is sufficient and necessary. In the unknown topology scenario we only consider
a single agent case. We first show that the task is impossible in the deterministic case if the agent has no memory.
Then we present a polynomial randomized algorithm. Finally in the model with memory we show that the number of steps
proportional to the number of edges of the graph is sufficient and necessary. In some sense we investigate how complex is
the problem of ``losing'' information about location (both physical and logical) for  different settings.

\end{abstract}

\section{Introduction}\label{sect:intro}

Let us consider the following problem. We have a group of mobile devices/sensors called \emph{agents} performing a task
on a given area. The task could be for example collecting/detecting some valuable resource, mounting detectors or installing mines.
In all aforementioned examples the owner of the system may want to hide the location of the agents against an adversary
observing the terrain from the satellite or a drone. That is, location of the devices may leak sensitive information to
the adversary. If we assume that the adversary's surveillance of the terrain is permanent and precise then clearly no
information can be concealed. Hence in our scenario there are periods of time when the adversary 
cannot observe the system during which the actual tasks are performed. Upon the approaching adversary, the devices launch
an algorithm to \textit{hide} their \textit{location}, i.e. they change their positions to mislead the observer. Clearly
in many real life  scenarios the additional movement dedicated for hiding their previous position should be possibly short
for the sake of saving energy and time. It  is also clear that the devices may want to return to their original
positions in order to resume their activities (When the adversary stops surveillance). On the other hand it is intuitively clear that a very short move may be
not sufficient  for ``losing'' the information about the starting  positions.

The outlined description is an intuitive  motivation for the research presented in our paper. \textbf{Exactly} the same problem can be however 
considered in many other settings  when we demand  ,,quick`` reconfiguration of a system such that the observable configuration should say 
possibly small about the initial one (i.e., agents exploring a network). For that reason we decided to use quite  general 
mathematical model.  Namely  the agents are placed
in vertices of a graph can move only through edges (single edge in a single round). Our aim is to design an algorithm
that governs the movement of the agents to change their initial location to the final location in such a way that the
adversary given the final assignment of agents cannot learn their initial positions. At hand one can point  the following
strategy - every agent chooses independently  at random vertex on the graph and moves to this location. 
Clearly (but informally) the new  location does not reveal  any information about the initial position. That is, the
initial locations of agents are perfectly hidden from the adversary. Note however that the same effect can be obtained
if all agents goes to a single, fixed in advance vertex. In this case, again the final position is stochastically
independent on the initial position. These strategies requires however that agents know the topology of the graph.

Intuitively, similar  effect can be achieved if each agent starts a random walk
and stop in a vertex  after some number  of steps. It this approach, the knowledge of the graph is \textbf{not} necessary, however
one can easily see that state after any number of steps reveals  some knowledge about the initial position (at least in some graphs).
Moreover in this strategy we need randomization. To summarize, there are many different methods for hiding the initial
location.  It turns that the possible solutions and their efficiency  depends greatly on the details of the assumed
model -- if the graph is known to the agents, what memory is available to each agent, if the agents can communicate and
if the agents have access to a source of random bits. Our paper is devoted to formalizing this problem and discussing
its variants in chosen setting. 

\subsubsection*{Organization of the Paper}
 In Section \ref{sect:model} we describe formally the problem and
the formal model. Section \ref{sect:results}  summarizes the obtained results discussed in the paper.
The most important related work is mentioned  in Section \ref{sect:rel}. In  Section \ref{sect:known} we present results
for the model wherein stations know the topology of the graph representing the terrain. We show both optimal algorithms
as well as respective lower bounds. The case with unknown topology is discussed in Section \ref{sect:unknown}.
We conclude in Section~\ref{sect:conclusion}. Some basic facts and definitions from Information Theory 
are recalled in Appendix~\ref{sect:apend}. In Appendix~\ref{sect:apendMC} we remind some definitions
and properties of Markov chains and random walks.

\section{Model}
\label{sect:model} 

\subsubsection*{The Agents in the Network}
We model the network as a simple, undirected, connected graph (i.e. a graph with no loops and multiple edges) with $n$
vertices, $m$ edges and diameter $D$. The nodes of the graph are uniquely labeled with numbers $\{1,2,\dots,n\}$.
We have also $k\geq 1$ agents representing mobile devices. Time is divided into synchronous rounds. At the beginning of
each round each agent is located in a single vertex. In each round the agent can change its position to any of
neighboring vertices.  We allow many agents to be in a single vertex in the same round. The agents need to locally
distinguish the edges in order to navigate in the graph hence we assume that the edges outgoing from a node with degree
$d$ are uniquely labeled with numbers $\{1,2,\dots,d\}$. We assume no correlation between the labels on two endpoints of
any node. A graph with such a labeling is sometimes called \emph{port-labeled}.

When an agent is located in a vertex we assume that it has access to the degree of the node and possibly the value or
the estimate of $n$ and to some internal memory sufficient for local computations it performs. 
In our paper we consider various models of mobile agents depending on the resources at their disposal. This will involve 
settings where the devices have or have not an access to a source of random bits and they are given a priori the topology
of the network or they have no such knowledge. In the latter case we will consider two different scenarios depending on
whether the agent has an access to operational memory that remains intact when it traverses an edge or its memory is very
limited and does not allow to store any information about the network gathered while it moves from one vertex to another. 

Our primary motivation is the problem of physical hiding of mobile devices performing their tasks in some terrain.
Nevertheless, our work aims for formalizing the problem of losing information on positions initially occupied by the agents
in some given network. Thus, we focus on proposing a theoretical model related to the logical topology that can be described
by a connected graph representing the underlying network consisting of the nodes corresponding to positions where mobile
devices can be placed and the edges the agents can traverse for moving to another location.

\subsubsection*{Model of the Adversary}

From the adversary's point of view, the agents are indistinguishable whilst the nodes of the underlying graph are labeled.
The assumption on indistinguishability of the devices is adequate for the system with 
very similar objects.

Thus the state of the system in a given round $t$ can be seen as a graph $G$ and a function $n_t(v)$ denoting the number
of agents located at node $v$.
Let $X_t$ for $t \in \{ 0,1, \ldots \}$ represent the state of the network at the beginning of the $t$-th round. 

We assume that in round $0$ the agents complete (or interrupt due to approaching adversary) their actual tasks and run
hiding algorithm $\AAA$ that takes $T$ rounds. Just after the round $T$ the adversary is given  the final state $X_T$ and, roughly speaking, its aim is to learn as much as possible about the initial state $X_0$. That is, the adversary gets an access to a single configuration (representing a single view of a system).

Note that the adversary may have some \textit{a priori} knowledge that is modeled as a distribution
of $X_0$ (possibly unknown to agents). In randomized settings the adversary has no information about agents' local random number generators. On the
other hand, the aim of agents is to make learning  the adversary $X_0$ from $X_T$ impossible  for \textbf{any}
initial state (or distributions of states)\footnote{In that sense we consider the worst case scenario that implies
strongest security guarantees}. Moreover the number of rounds $T$ should be as small as possible (we need to hide the
location quickly). We also consider \textit{energy complexity} understood as the maximal number of moves (i.e. moving to
a neighboring vertex) over all agents in the execution of $\AAA$. Such definition follows from the fact that we need to
have all agents working and consequently we need to protect the most loaded agent against running out of batteries.
As we shall see, in all cases considered in this paper the energy complexity is very closely related to the time of
getting to the ``safe'' configuration by all devices, namely it is asymptotically equal $T$.

\subsubsection*{Security Measures}

Let $X_0$ be a random variable representing the knowledge of the adversary about the initial state and let $X_T$
represent the final configuration of the devices after executing algorithm $\AAA$. We would like to define a measure of
efficiency of algorithm $\AAA$ in terms of location hiding. In the case of problems based on ``losing'' knowledge,
hiding information, there is no single, commonly accepted definition. For example the problem of finding adequate
(for practical reasons) security measure in the case of anonymous communication has been discussed for a long time.
These discussion are reflected in dozens of papers including \cite{ANON1} and \cite{ANON2}. Nevertheless the good
security measure needs to estimate  ``how much information'' about $X_0$ is in $X_T$. 

Let $X \sim \mathcal{L}$ be a random variable following probability distribution $\mathcal{L}$. We denote by $\E{X}$ 
the expected value of $X$. By $\Unif{A}$ we denote the uniform probability distribution over the set $A$ and by $\Geo{p}$
the geometric distribution with parameter $p$. An event $E$ occurs with high probability (w.h.p.) if for an arbitrary
constant $\alpha > 0$ we have $\Pr[E] \geq 1 - \BigO{n^{-\alpha}}$.
Let $\HH{X}$ denotes the entropy of the random variable $X$. Similarly  $\HH{X|Y}$ denotes conditional entropy and
$\MutInf{X}{Y}$ denotes mutual information. All that notations and definitions are recalled in the
Appendix \ref{sect:apend}. Let us also notice, that through the paper we will restrict only to simple, connected and
undirected graphs and unless otherwise stated, $G$ will always denote such graph.

Our definition is based on the following notion of \textit{normalized mutual information}, also known as
\textit{uncertainity coefficient} (see e.g. \cite{NumRecipes}, chapter 14.7.4)
$$ 
	\UC{X}{Y} = \frac{\MutInf{X}{Y}}{\HH{X}}~.
$$
From the definition of mutual information it follows that $\UC{X}{Y} = 1 - \frac{\HH{X|Y}}{\HH{X}}$
and $0 \leq \UC{X}{Y} \leq 1$. The uncertainity coefficient $\UC{X}{Y}$ takes the value 0 if there is no association
between $X$ and $Y$ and the value 1 if $Y$ contains all information about $X$. Intuitively, it tells us what portion
of information about $X$ is revealed by the knowledge of $Y$. $\HH{X} = 0$ implies $\HH{X|Y} = 0$ and we may use the
convention that $\UC{X}{Y} = 0$ in that case. Indeed, in such case we have stochastic independence between $X$ and $Y$
and the interpretation in terms of information hiding can be based on the simple observation that $\HH{X} = 0$ means
that there is nothing to reveal about $X$ (as we have full knowledge of $X$) and $Y$ does not give any extra information.

\begin{definition}
	\label{def:hiding}
The algorithm $\AAA$ is $\epsilon$-hiding if for any distribution of the initial configuration $X_0$ with non-zero entropy
(i.e. $\HH{X_0} > 0$) and for \textbf{any} graph $G$ representing the underlying network
\begin{equation}
	\label{eq:epsilonHiding}
	\UC{X_0}{X_T} = \frac{\MutInf{X_0}{X_T}}{\HH{X_0}} \leq \epsilon~,
\end{equation}
where $X_T$ is the state just after the execution of the algorithm $\AAA$.
\end{definition}

\begin{definition}
\label{def:wellHiding}
We say that the algorithm $\AAA$ is 
\begin{itemize}
	\item well hiding if it is $\epsilon$-hiding for some $\epsilon \in \SmallO{1}$;
	\item perfectly hiding if it is $0$-hiding.
\end{itemize}
\end{definition}

Intuitively, this definition says that the algorithm works well if the knowledge of the final state
reveals only a very small fraction of the total information about the initial configuration 
regardless of the distribution of initial placement of the devices.
Let us mention that theses definitions state that any hiding algorithm should work well regardless
of the network topology. If an algorithm $\AAA$ is $\epsilon$-hiding, then for \textbf{any}
simple connected graph $G$ and for \textbf{any} probability distribution of agents' initial positions $X_0$
in that graph with non-zero entropy the final configuration $X_T$ after $\AAA$ terminates should fulfill
Equation \ref{eq:epsilonHiding}. Notice also that there are some cases when it is not feasible to hide the
initial location in a given graph. Assuming the adversary knows the agents' initial distribution $X_0$,
$\HH{X_0} = 0$ means that the agents with probability 1 are initially placed in some fixed location which
is known to the adversary, hence there is no possibility to hide them. In particular, this is the case when
the graph representing the network has only one vertex (it can be a model of system with exactly one state).
All mobile devices must be then initially located in that vertex and no hiding algorithm exists for this setting. 

The general idea of location hiding algorithms is depicted in Fig. \ref{fig:locHidingAlg}.
The agents are initially placed in some vertex of the graph $G$ (Fig. \ref{fig:locHidingAlg:0})
according to some known distribution of their initial configuration $X_0$. 
In each step every agent located in some vertex $v_i$ can move along an edge incident $v_i$ or stay in $v_i$.
After $T$ steps the algorithm terminates resulting in some final configuration $X_T$ of agents' positions
(Fig. \ref{fig:locHidingAlg:3}). Our goal is to ensure that any adversary observing the positions $X_T$ of mobile
devices after execution of an location hiding algorithm can infer as small information as possible about
their actual initial placement $X_0$, regardless of $G$ and the distribution of $X_0$.
\begin{figure}[!ht]
	\centering
	\begin{subfigure}{.5\linewidth}
		\centering
		\includegraphics[width=0.95\linewidth]{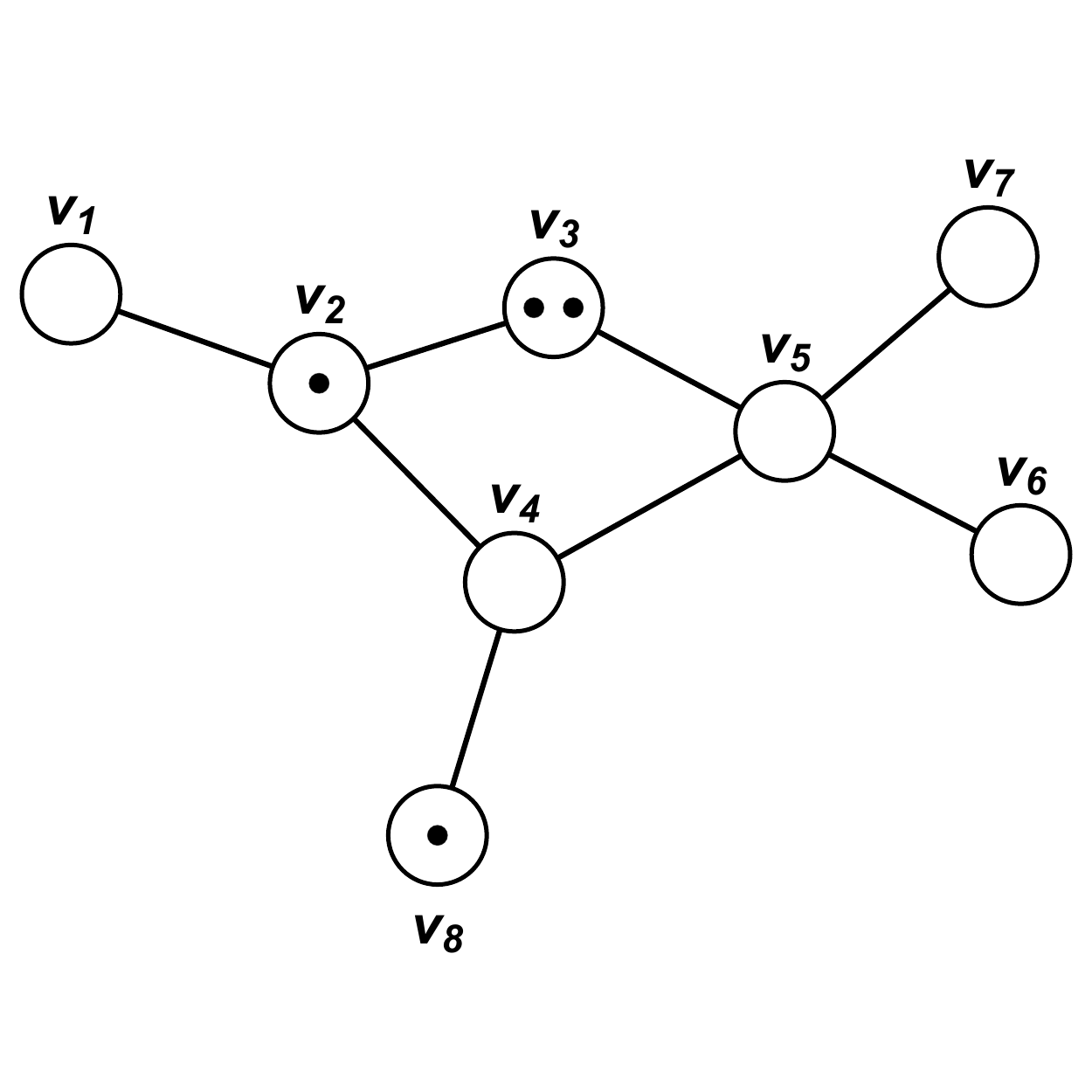}
		\caption{$t = 0$}
		\label{fig:locHidingAlg:0}
	\end{subfigure}%
	\begin{subfigure}{.5\linewidth}
		\centering
		\includegraphics[width=0.95\linewidth]{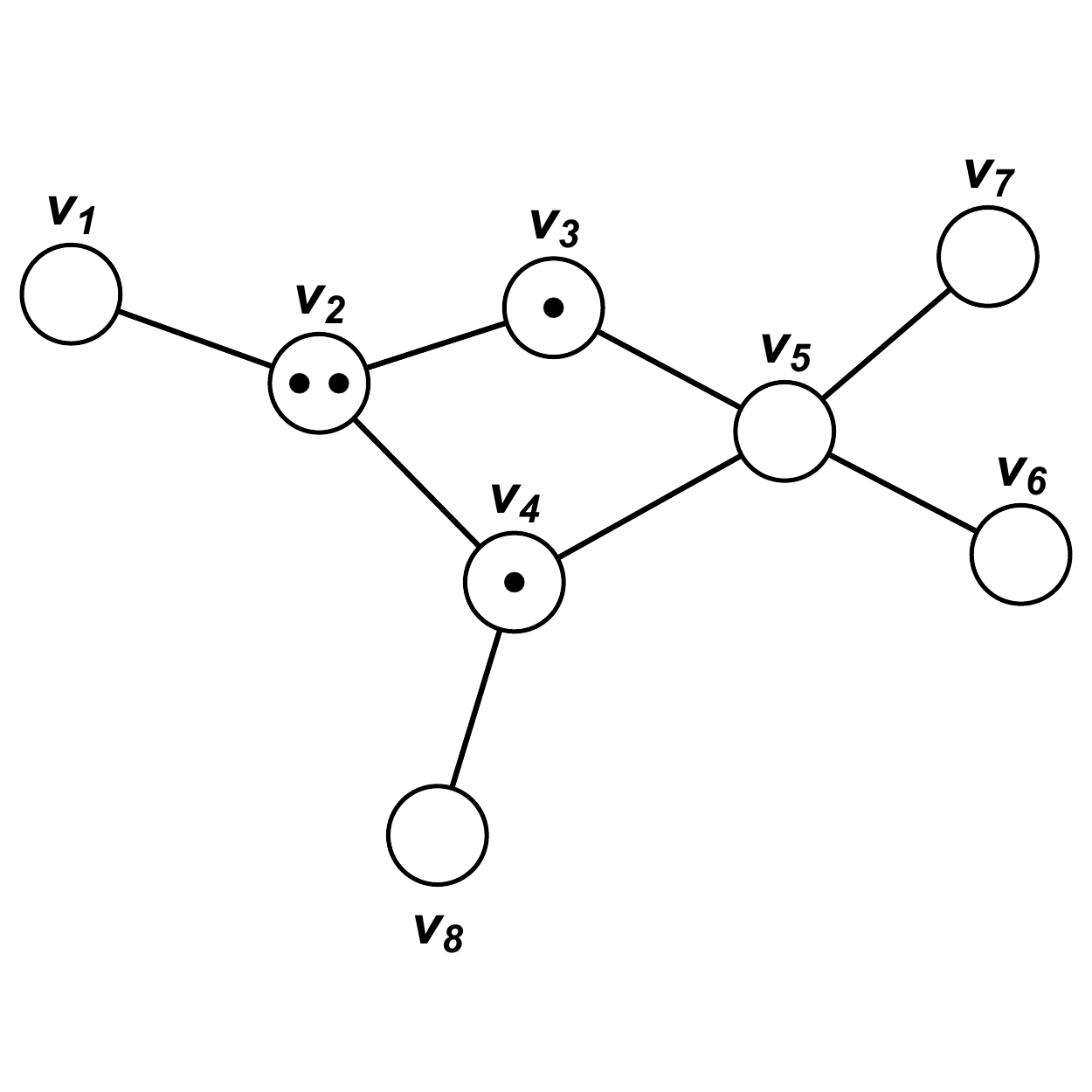}
		\caption{$t = 1$}
		\label{fig:locHidingAlg:1}
	\end{subfigure}%

	\begin{subfigure}{.5\linewidth}
		\centering
		\includegraphics[width=0.95\linewidth]{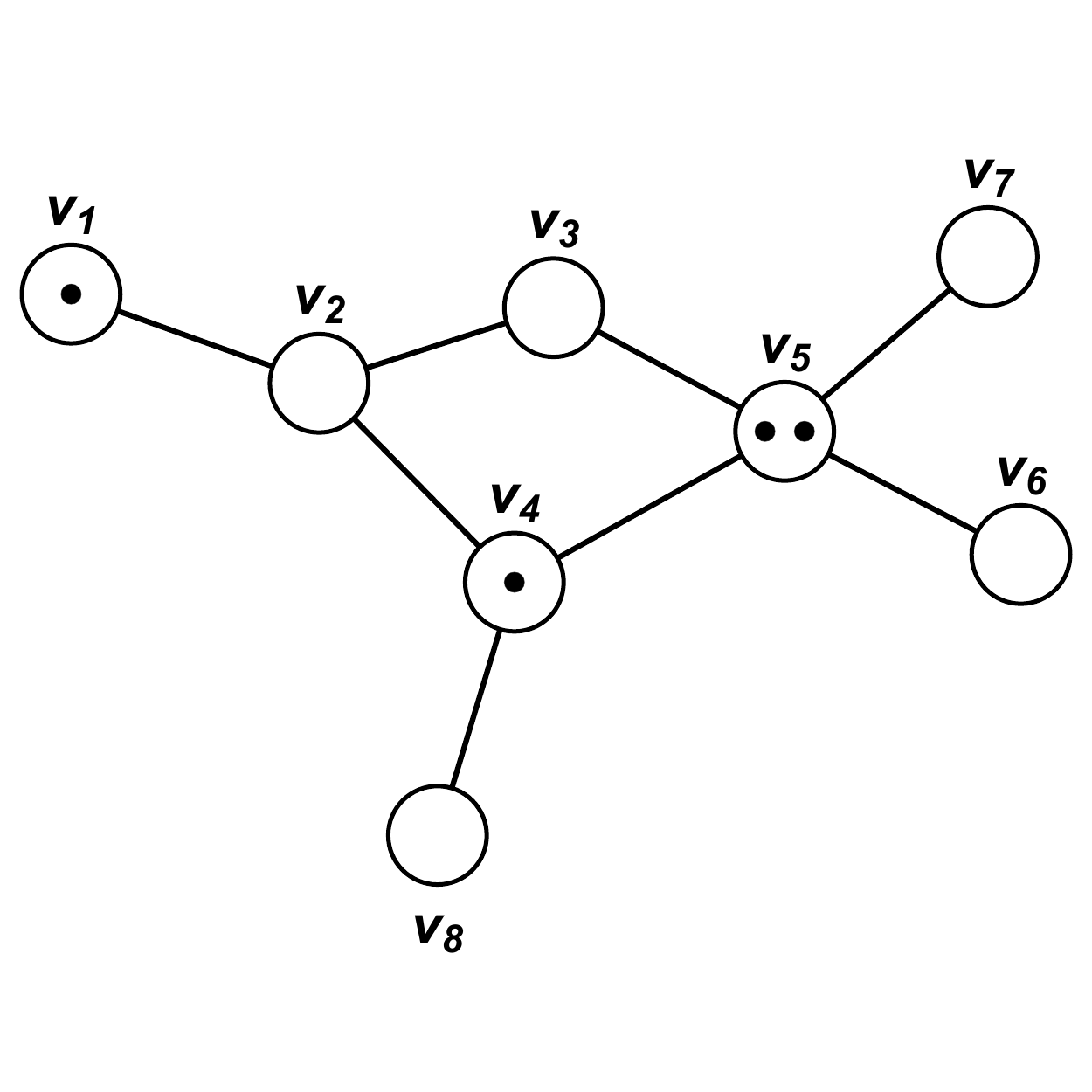}
		\caption{$t = 2$}
		\label{fig:locHidingAlg:2}
	\end{subfigure}%
	\begin{subfigure}{.5\linewidth}
		\centering
		\includegraphics[width=0.95\linewidth]{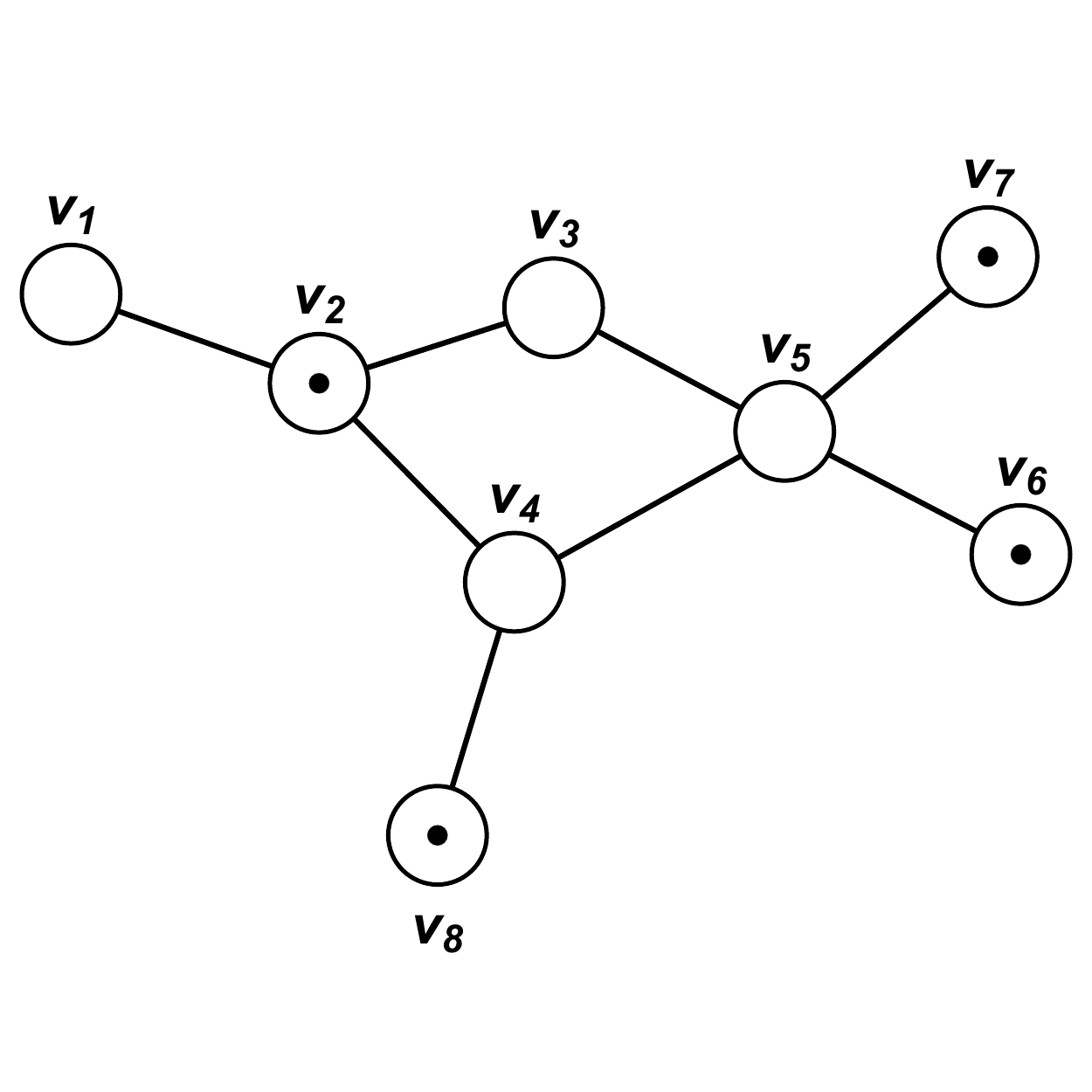}
		\caption{$t = 3 = T$}
		\label{fig:locHidingAlg:3}
	\end{subfigure}
	\caption{Position of $k = 4$ indistinguishable agents (represented by black dots) in consecutive steps
	  of sample execution of location hiding algorithm in a network with $n = 8$ nodes.}
	\label{fig:locHidingAlg}
\end{figure}
\section{Our Results}
\label{sect:results}
Most of our results apply to the single-agent case. We first show that if the topology of the graph is known then any well-hiding algorithm in a graph with $n$ nodes $m$ edges and diameter $D$ needs $\Omega(D)$ steps, and there exists a perfectly hiding algorithm that needs $O(D)$ steps. Then we  show how to generalize this result to multi-agent case.

Secondly we focus on the model where the topology of the network is unknown. We first show that in the deterministic model with no memory there exists no well-hiding algorithm. On the other hand in the randomized setting with no memory we show a well-hiding algorithm whose expected running time is $\widetilde{O}(n^3)$ with high probability.
Finally if the agents are given unlimited memory then $\Theta(m)$ is necessary  for perfectly hiding and sufficient  for well-hiding algorithms.

Our results are summarized in Table~\ref{tab1}.
\begin{table}
	\caption{Overview of our results}\label{tab1}
	\centering
	\begin{tabularx}{\textwidth}{XXXX}
		\toprule
		& & Deterministic & Randomized \\\midrule
		Known topology & & $\Theta(D)$ (Thm.\ref{thm:known}) & $\Theta(D)$ (Thm.\ref{thm:known}) \\\midrule
		\multirow{2}{1pt}{Unknown\\topology} & No memory & impossible (Thm.~\ref{thm:detNoMemo}) & $\widetilde{O}(n^3)$ w.h.p. (Thm.~\ref{thm:unknowntopoMemorylessRnd}) \\ \cline{2-4}
		 & Unlimited Memory & $\Theta(m)$ (Thm.~\ref{thm:DFS},~\ref{thm:unknownMemory}) & $\Theta(m)$ (Thm.~\ref{thm:DFS},~\ref{thm:unknownMemory}) \\
		\bottomrule
	\end{tabularx}
\end{table}

\section{Previous and Related Work}\label{sect:rel}

The problems of security and privacy protection in distributed systems and mobile sensor networks have recently
received a lot of attention. Various security aspects of such systems have been extensively discussed and
a lot of novel solutions for some practical settings have been proposed over the last years.
One of the major examples is the problem of designing routing protocols for wireless ad hoc networks
which can hide the network topology from potential external and internal adversaries.
(see e.g. \cite{Zhang:2014:TOHIP}, \cite{Pani:2016:TSOR} and the references therein). The goal of such protocols is to
find reliable routes between the source and destination nodes in the network which are as short as possible,
reducing exposure of the network topology to malicious nodes by data transmitted in the packets. This will prevent
adversaries (at least to some extent) for launching some kinds of attacks requiring the knowledge of the network
topology which may be particularly harmful for the whole network and the tasks performed. It is desired that such
algorithms should meet some security requirements like confidentiality, integrity or authenticity.
This can be achieved e.g. by ensuring that routing packets do not carry any information of the discovered routes that
can be exploited by an adversary, as proposed in \cite{Pani:2016:TSOR}.

Another important line of research is the issue of assuring privacy of the users of mobile applications
utilizing the information of their location, like e.g. in activity-based social networks, where users share the
information about location-based activities they perform (see e.g. \cite{Pham:2014:Activity}). The exact knowledge
of such location data by the service providers in some cases may rise serious concerns about users' privacy.
Because popularity and usage of such applications is still increasing, there is a need of designing the algorithms
which allows for hiding the exact users' location without significant limitations of the services. 

However, in some cases the performance of proposed protocols is evaluated only experimentally and the 
discussion of security properties of the solutions is somewhat informal, without referring to any theoretical
model (cf. \cite{Pani:2016:TSOR}, \cite{Pham:2014:Activity}).

To the best of our knowledge, there is no rigid and formal analysis on the problem of location hiding in graphs
and it has never been studied before in the context considered in this paper.
The problems of ensuring security and privacy in distributed systems mentioned above are similar
to our only to a certain extent. The aim of our approach is to propose a formal and general model of hiding the
positions of a set of mobile agents from an external observer and consider its basic properties and limitations.
However, the problem considered by us is closely related to some of the most fundamental agent-based graph problems.

First of all observe the relation to the exploration that comes from the global nature of our problem. Clearly if the
agent has at least logarithmic memory then we can use algorithms for graph exploration. Indeed, since the graph
is labeled, it is sufficient to explore the graph and move to a vertex with minimum ID. Hence the vast body of
literature about exploration in various models applies to our problem. In particular there exist polynomial
deterministic algorithms (using Universal Sequences) that need only logarithmic memory~\cite{AleliunasKLLR79,Reingold08}.

In the randomized setting, location hiding becomes related to the problem of reaching the stationary distribution
by a Markov Chain (Mixing Time) as well as visiting all the states (Cover Time). 
The cover time of a random walk is the expected number of steps that are needed by a random walk to visit all the
vertices. It is known that for a (unbiased) random walk, the cover time is between $\Omega(n\log n)$~\cite{FeigeLower}
and $O(n^3)$~\cite{FeigeUpper}. There exist biased random walks that achieve worst-case cover time of
$\widetilde{O}(n^2)$~\cite{NonakaOSY10} however in order to implement such walks the agent is required to have access to
some memory to acquire information necessary to compute the transition probabilities. It has been recently shown that in
some graphs multiple random walks are faster than a single one~\cite{AlonAKKLT11,EfremenkoR09,ElsasserS11}.
Another  interesting line of work is deriving biased random walks (~\cite{BoydDX04,Diaconis2}).
 
\section{Location Hiding for Known Topology}
\label{sect:known}

Let us first focus on the setting where the topology of the underlying network 
is known to the agents and consider one of the simplest possible protocols, namely
every mobile agent goes from its initial positions to some fixed vertex $v^{*} \in V$
(this is possible, because in the considered scenario the vertices in the graph have unique
labels known to all the agents).
One can easily see that this simple protocol is perfectly hiding.
Indeed, regardless of the distribution of the agents' initial placement,
after executing the protocol all devices are always located in the same
vertex known in advance. Hence, $X_{T}$ and $X_0$ are independent and
$\MutInf{X_0}{X_T}$ = 0 (and therefore $\UC{X_0}{X_T}$ = 0, as required).
But this approach leads to the worst case time and energy complexity for a single device
of order $\Theta(D)$, where $D$ is the graph diameter. Appropriate selection of the
vertex $v^{*}$ as an element of the graph center can reduce the worst case complexity
only by a constant, but it does not change its order. The natural question that arises in this context is whether there exist a perfectly hiding
(or at least well hiding) protocol that requires asymptotically smaller number of rounds
for ensuring privacy than the simple deterministic protocol discussed above.
In general, we are interested in determining the minimal number of steps required by any
location hiding protocol in considered scenarios for ensuring a given level of security
(in terms of the amount of information being revealed) for arbitrary distribution of initial
configuration of the agents and for arbitrary underlying network.

\subsection{Single Agent Scenario}

Let us consider the simple scenario where there is only one mobile device in the network
located in some vertex $v \in V$ 
 according to some known
probability distribution $\LLL$ over the set of vertices. Assume that the network topology
is known to the agent. Our goal is to find the lower bound on the number of steps that
each well hiding protocol requires to hide the original location of the device in this scenario
for arbitrary graph $G$ and initial distribution $\LLL$.

We will start with a general lemma. We would like to show that if within $t$ steps the sets of vertices visited by the
algorithm starting from two different vertices are disjoint with significant probability, then the algorithm is not
well hiding within time $t$. Intuitively it is because it is easy to deduce which was the starting position.
This result formalizes the following lemma.

\begin{lemma}
	\label{lem:known}
	Let $\AAA$ be any hiding algorithm and $G = (V,E)$ be an arbitrary graph. Suppose that for some $t > 0$ and some
	positive constant $\gamma$ there exist two distinct vertices $u, v \in V$ such that with probability at least 
	$1/2 + \gamma$ the following property holds: sets $V_1$ and $V_2$ of vertices reachable after the execution of $t$
	steps of $\AAA$ when starting from $u$ and $v$, respectively, are disjoint. Then $\AAA$ is not well hiding in time $t$.
\end{lemma}
\begin{proof}
  Fix an arbitrary graph $G = (V,E)$ with $|V| = n$ and hiding algorithm $\AAA$. Let $u, v \in V$ be two vertices such
	that the sets $V_1$ and $V_2$ of possible location of the agent after performing $t$ steps of $\AAA$ when starting in
	$u$ and $v$, respectively, are disjoint with probability at least $1/2 + \gamma$ for some constant $\gamma > 0$
	regardless of the starting point, i.e. $\Pr[\xi_V = 1] \geq 1/2 + \gamma$, where $\xi_V$ is an indicator random variable
	of the event $V_1 \cap V_2 = \emptyset$.
	Consider the following two-point distribution $\LLL$ 	of the agents' initial location $X_0$: $\Pr[X_0 = u] =
	\Pr[X_0 = v] = 1/2$, $\Pr[X_0 = w] = 0$ for $w \in V \setminus \{u, v\}$. 
	We will prove that such $\AAA$ does not ensure that the initial position $X_0$ of the device is well hidden
	at time $t$ when $X_0 \sim \LLL$.
	
	It is easy to calculate that $\HH{X_0} = 1$. Hence $\UC{X_0}{X_t} = \MutInf{X_0}{X_t}$ and it suffices to show that
	the mutual information $\MutInf{X_0}{X_t} \geq \eta > 0$ for some positive constant $\eta$.
	This in turn is equivalent to $\HH{X_0|X_t} \leq 1 - \eta$, as follows from Fact~\ref{eq:MutInf2}. 

	It is clear that for $y \in V_1$
	\begin{equation}
	  \label{eq:CondProbLem1}
		\Pr[X_0 = u | X_t = y] \geq \Pr[X_0 = u, \xi_V = 1 | X_t = y] \geq 1/2 + \gamma 
	\end{equation}
	and the same holds after replacing $u$ with $v$. Moreover $\Pr[X_0 = v | X_t = y] = 1 - \Pr[X_0 = u | X_t = y]$.
	Denoting $\Pr[X_0 = u | X_t = y]$ by $p_{u|y}$ we have
	\begin{align*}
		\HH{X_0|X_t} & = -\sum_{y \in V} \Pr[X_t=y] \sum_{x \in V} \Pr[X_0=x|X_t=y] \log(\Pr[X_0=x|X_t=y])\\
		             & = -\sum_{y \in V} \Pr[X_Tt=y] \left(p_{u|y} \log(p_{u|y}) + (1-p_{u|y}) \log(1 -p_{u|y})\right)~.
	\end{align*}
	Let us consider the function $f(p) = -(p\log(p) + (1-p)\log(1-p))$ for $p \in (0,1)$. It is easy to see that
	$\lim_{p \rightarrow 0} f(p) = \lim_{p \rightarrow 1} f(p) = 0$ and $f(p)$ has its unique maximum on the interval (0,1) 
	equal to 1 at $p = 1/2$.
	From \ref{eq:CondProbLem1} we have $(\forall{y \in V_1})(p_{u|y} \geq 1/2 + \gamma)$ and 
	$(\forall{y \in V_2})(p_{u|y} \leq 1/2 - \gamma)$. Therefore, there exists some positive constant
	$\eta$ such that $f(p_{u|y}) \leq 1-\eta$ 
	From the definition of the sets $V_1$ and $V_2$ we also have $\Pr[X_t = y \notin V_1 \cup V_2] = 0$.
	Using these facts the conditional entropy $\HH{X_0|X_t}$ can be rewritten as
	\begin{eqnarray*}
		\HH{X_0|X_t} &    = & \sum_{y \in V_1} \Pr[X_t=y] f(p_{u|y}) + \sum_{y \in V_2} \Pr[X_t=y] f(p_{u|y}) \\
								 & \leq & (1-\eta) \Pr[X_t \in V_1 \cup V_2] = 1 - \eta
	\end{eqnarray*}
	for some constant $\eta > 0$, as required.
	Hence, the lemma is proved.
\end{proof}

From the lemma we get as a corollary the lower bound of $\Omega(D)$ on the expected number of steps needed by any
well hiding algorithm in the model with known topology. Note that the lower bound matches the simple $O(D)$ upper bound.

\begin{theorem}
	\label{thm:known}
	For a single agent and known network topology and for an arbitrary graph $G$ there exists a distribution
  $\LLL$ of agent's initial position such that  any well hiding algorithm $\AAA$ needs to perform at least
	$\left\lfloor D/2 \right\rfloor$ steps with probability $c \geq 1/2 - \SmallO{1}$,
	where $D$ is the diameter of $G$.
\end{theorem}
\begin{proof}
	We will show that for each graph $G$ there exist a distribution of the initial state of the mobile agent such that
	each well hiding algorithm $\AAA$ needs at least $\left\lfloor D/2 \right\rfloor$ rounds with some probability
	$c \geq 1/2 - \SmallO{1}$. 
	
	Fix an arbitrary graph $G = (V,E)$ with $|V| = n$. Let $u, v \in V$ be two vertices such that $d(u,v) = D$.
	Denote by $\delta = \left\lfloor D/2 \right\rfloor$ and consider the following two-point distribution $\LLL$
	of the agents' initial location $X_0$: $\Pr[X_0 = u] = \Pr[X_0 = v] = 1/2$, $\Pr[X_0 = w] = 0$ for $w \in V \setminus \{u, v\}$.
	Suppose that some hiding algorithm $\AAA$ terminates with probability at least $1/2 + \gamma$ for some constant
	$\gamma > 0$ after $T < \delta$ steps regardless of the starting point, i.e. $\Pr[T < \delta] \geq 1/2 + \gamma$.
	
	Obviously there is no $z \in V$ such that $d(u,z) < \delta$ and $d(v,z) < \delta$
	(if so, $D = d(u, v) < 2 \left\lfloor D/2\right\rfloor \leq D$ and we will get a contradiction).
	Let us define $B(u,\delta) = \{y \in V \colon d(u,y) < \delta\}$ and $B(v,\delta) = \{y \in V \colon d(v,y) < \delta\}$.
	It is clear that $B(u,\delta) \cap B(v,\delta) = \emptyset$.
	From the assumptions on the algorithm' running time with probability at least $1/2 + \gamma$ the sets $V_1$
	and $V_2$ of vertices reachable from $u$ and $v$, respectively, fulfills $V_1 \subseteq B(u,\delta)$ and
	$V_2 \subseteq B(v,\delta)$, therefore they are disjoint.
	Hence it suffices to apply the results from Lemma \ref{lem:known} to complete the proof.
\end{proof}

\subsection{Location Hiding for \texorpdfstring{$k$}{k} Agents and Known Network Topology}

Let us recall that the energy complexity of an algorithm $\AAA$
in the multi-agent setting is defined as the maximal distance covered (i.e. number of moves) in the execution of $\AAA$
over all agents. 

In the general scenario considered in this section a similar result holds as for the single-agent case. Namely, each algorithm 
which ensures the well hiding property regardless of the distribution of agents' initial placement requires in the worst case
$\Omega(D)$ rounds. 
\begin{lemma}
	\label{lem:knownMulti}
  For known network topology and $k > 1$ indistinguishable agents initially placed according to some arbitrary
	distribution $\LLL$, any well hiding algorithm $\AAA$ for an arbitrary graph $G$ representing the underlying network
	has energy complexity at least $\left\lfloor D/2 \right\rfloor$ with probability $c \geq 1/2 - \SmallO{1}$,
	where $D$ is the diameter of $G$.
\end{lemma}

The proof of the Lemma \ref{lem:knownMulti} proceeds in the same vein as in Theorem \ref{thm:known}. We choose two
vertices $u, v$ in distance $D$ and put all agents with probability $1/2$ in any of these vertices. Denoting by
$T_i$, $1 \leq i \leq k$, the number of steps performed by the agent $i$ and by $T = \max_{1 \leq i \leq k} T_i$
the energy complexity of the algorithm we consider a hiding algorithm $\AAA$ such that $\Pr[T < \delta] \geq 1/2 + \gamma$
for $\delta = \left\lfloor D/2 \right\rfloor$ and some positive constant $\gamma > 0$.
The only difference is that instead of the sets $B(u,\delta) = \{y \in V \colon d(u,y) < \delta\}$ and
$B(v,\delta) = \{y \in V \colon d(v,y) <\delta\}$ itself we consider the subsets $\states_1$ and $\states_2$
of the state space consisting of such states that all of the agents are located only in the vertices from the set
$B(u,\delta)$ or $B(v,\delta)$, respectively. Similar calculations as previously led to the conclusion that any such
algorithm cannot ensure the well hiding property.

Let us notice that the natural definition of the energy complexity in the multi-agent case as the maximum number of
moves over all agents performed in the execution of the algorithm allows us for direct translation of the result
from the single-device setting, as presented above. Nevertheless, another interesting questions are what happens
if we define the energy complexity as the total or average number of moves over all devices.

\section{Location Hiding for Unknown Topology}
\label{sect:unknown}

In this section we consider the case when the agents do not know the topology of the network.
  
\subsection{No Memory}
If no memory and no information about the topology is available but the agent is given access to a source of randomness,
the agent can perform a random walk in order to conceal the information about its starting position. But the agent would
not know when to stop the walk. If in each step it would choose to terminate with probability being a function of the 
degree of the current node, one could easily construct an example in which the agent would not move far from its
original position (with respect to the whole size of the network). Hence in this section we assume that the size of the
network is known. Then interestingly the problem becomes feasible. Consider the following algorithm $\AAA(q)$: in each
step we terminate with probability $q$ (roughly $n^{-3}$), and with probability $1-q$ we make one step of a lazy random
walk. We will choose the appropriate $q$ later. Let us point out that letting the random walk to stay in the current
vertex in each step with some fixed constant probability is important for ensuring the aperiodicity of the Markov
process (see e.g. \cite{LevinPeresWilmerMCMT} and \cite{lovasz1993random}). Otherwise we can easily provide an example
where such algorithm does not guarantee the initial position will be hidden. Namely, consider an arbitrary bipartite
graph and any initial distribution s.t. the agent starts with some fixed constant probability either in some fixed
\emph{black} or \emph{white} vertex. Then, assuming the adversary is aware only of the algorithm' running time (i.e. the
number of steps the agent performed), when observing agents' position after $T$ steps the adversary can with probability
1 identify its initial position depending on $T$ is even or odd. Nevertheless, the probability of remaining in a given
vertex can be set to arbitrary constant $0 < c < 1$. For the purposes of analysis we have chosen $c = 1/2$ which leads
to the classical definition of lazy random walk (see Definition \ref{def:rw} and Fact \ref{fact:RWMixing}
in Appendix \ref{sect:apendMC}).

\begin{algorithm}
	\label{alg:unknownMemorylessRand}
	\textbf{Algorithm $\AAA(q)$} [randomization, no memory, no topology, knowledge of $n$]\\
	In each round: 
	\begin{enumerate}
		\item With probability $q$: terminate the algorithm.
		\item With probability $\frac12$: remain in the current vertex until the next round. 
		\item With probability $\frac12 - q$: move to a neighbor chosen uniformly at random.
	\end{enumerate}
\end{algorithm}

\begin{theorem}
\label{thm:unknowntopoMemorylessRnd}
The algorithm $\AAA(q)$ described above based on the random walk with termination probability
$q = \frac{f(n)}{n^{3} \log{h(n)}}$ for arbitrary fixed $f(n) = \SmallO{1}$ and
$h(n)$ = $\SmallOmega{\max\{n^{2}, \frac{1}{\HH{X_0}}\}}$ is well hiding for any graph $G$ and any distribution
of agent's initial location $X_0$.
\end{theorem}
\begin{proof}
Fix $\varepsilon > 0$. Let $\tmix{\varepsilon}$ denote the mixing time and $\pi$ the stationary distribution of the
random walk performed by the algorithm according to Definition \ref{def:MixingTime}. We will choose the exact value
for $\varepsilon$ later.

Let $X_0$ and $X_T$ denote the initial and final configuration, respectively.
In order to prove the lemma it suffices to show that $\frac{\HH{X_0|X_T}}{\HH{X_0}} = 1 - \SmallO{1}$,
what is equivalent to $\lim_{n \rightarrow \infty} \frac{\HH{X_0|X_T}}{\HH{X_0}} = 1$. This implies
that $\UC{X_0}{X_T} = \SmallO{1}$ as required by Definition \ref{def:wellHiding}.

Let $\xi_{\varepsilon} = \mathbf{1}[T > \tmix{\varepsilon}]$ be the indicator random variable defined as
$$
		\xi_{\varepsilon} = 
		\begin{cases}
				1, & \text{if } {T > \tmix{\varepsilon}}~,\\
				0, & \text{otherwise.}
		\end{cases}
$$
What we need to ensure is that the algorithm will stop with probability at least $1 - \SmallO{1}$ after 
$\tmix{\varepsilon}$ steps. It is clear that time $T$ when the station terminates
the execution of the algorithm follows $\Geo{q}$ distribution, hence
$\Pr[\xi_{\varepsilon} = 1] = (1-q)^{\tmix{\varepsilon}}$. Letting $q = f(n)/\tmix{\varepsilon}$ for
some $f(n) = \SmallO{1}$ implies $\Pr[\xi_{\varepsilon} = 1] = 1 - \SmallO{1}$, as required.

Let us consider $\HH{X_0|X_T}$. By Fact \ref{fact:condReduce} (see Appendix \ref{sect:apend}) we have
\begin{align*}
		\HH{X_0|X_T} & \geq \HH{X_0|X_T, \xi_{\varepsilon}} \geq \HH{X_0|X_T, \xi_{\varepsilon} = 1} \Pr[\xi_{\varepsilon} = 1]\\
		             & = (1 - \SmallO{1})\ \HH{X_0|X_T, \xi_{\varepsilon} = 1} \geq (1 - \SmallO{1})\ \HH{X_0|X_{\tmix{\varepsilon}}}~,
\end{align*}
where the last inequality follows directly from Fact \ref{fact:MCInfTh} in Appendix \ref{sect:apend}.

Let us introduce the following shorthand notation: $\Pr[X_0 = x] = p_{0}(x)$, $\Pr[X_{\tmix{\varepsilon}} = y] = p_{t}(y)$,
$\Pr[X_0 = x | X_{\tmix{\varepsilon}} = y] = p_{0}(x|y)$ and $\Pr[X_{\tmix{\varepsilon}} = y | X_0 = x] = p_{t}(y|x)$.
Obviously, $p_{0}(x|y) = \frac{p_{t}(y|x)}{p_{t}(y)} p_{0}(x)$.
Using this convention we can rewrite $\HH{X_0|X_{\tmix{\varepsilon}}}$ as
\begin{align}
		\HH{X_0|X_{\tmix{\varepsilon}}} & = - \sum_{y \in V} p_{t}(y) \sum_{x \in V} p_{0}(x|y) \log p_{0}(x|y) \label{eq:lem5p1} \\
		                             & = - \sum_{y \in V} \sum_{x \in V} p_{t}(y|x) p_{0}(x) \log p_{0}(x) 
																 - \sum_{y \in V} \sum_{x \in V} p_{t}(y|x) p_{0}(x) \log \frac{p_{t}(y|x)}{p_{t}(y)}~. \nonumber
\end{align}
	
The definition and properties of mixing time implies that there exist $\{\varepsilon_{y}^{(1)}\}_{y \in V}$ and
$\{\varepsilon_{y}^{(2)}\}_{y \in V}$ such that $\sum_{y \in V} \varepsilon_{y}^{(i)} \leq 2\varepsilon$ for $i \in \{1,2\}$ and
$\pi(y) - \varepsilon_{y}^{(1)} \leq p_{t}(y|x) \leq \pi(y) + \varepsilon_{y}^{(1)}$ and $\pi(y) - \varepsilon_{y}^{(2)}
\leq p_{t}(y) \leq \pi(y) + \varepsilon_{y}^{(2)}$. Let $\varepsilon_{y} = \max\{\varepsilon_{y}^{(1)}, \varepsilon_{y}^{(2)}\}$.
Because for any $y \in V$ $\pi(y) \geq 1/n^{2}$, letting $\varepsilon$ being arbitrary
$\varepsilon(n) = \SmallO{\min\{\frac{1}{n^2}, \HH{X_0}\}}$ will ensure that
$$
	\frac{p_{t}(y|x)}{p_{t}(y)} \leq \frac{\pi(y) + \varepsilon_{y}}{\pi(y) - \varepsilon_{y}} = 1 + \SmallO{\min\{1, \HH{X_0}\}}~. 
$$
Thus, the above relations allow us to find the lower bound on the conditional entropy $\HH{X_0|X_{\tmix{\varepsilon}}}$.
The first sum in \eqref{eq:lem5p1} gives us
\begin{align}
	- \sum_{y \in V} \sum_{x \in V} p_{t}(y|x) p_{0}(x) \log p_{0}(x) 
	          & \geq \sum_{y \in V} (\pi(y) - \varepsilon_{y}) \HH{X_0} \geq \HH{X_0}(1 - 4\varepsilon) \nonumber \\
						& = \HH{X_0}(1 - \SmallO{1})~, \label{eq:lem5p2}
\end{align}
whereas the second sum can be expressed as
\begin{align*}
		- \sum_{y \in V} \sum_{x \in V} p_{t}(y|x) p_{0}(x) \log \frac{p_{t}(y|x)}{p_{t}(y)}
		   & = - \sum_{x \in V} p_{0}(x) \sum_{y \in V} p_{t}(y|x) \log \frac{p_{t}(y|x)}{p_{t}(y)} \\
			 & = - \sum_{x \in V} p_{0}(x) \DD{p_{t}(y|x)}{p_{t}(y)}~. 
\end{align*}
Applying the upper bound on the relative entropy from Fact \ref{thm:RelEnt2} we get
\begin{align}
  \sum_{x \in V} p_{0}(x) \DD{p_{t}(y|x)}{p_{t}(y)}
	    & \leq \sum_{x \in V} p_{0}(x) \frac{1}{\ln 2} \left(\sum_{y \in V} \frac{(p_{t}(y|x))^{2}}{p_{t}(y)} - 1 \right) \label{eq:lem5p3}\\
	    & \leq \frac{1}{\ln 2} \sum_{x \in V} p_{0}(x) \sum_{y \in V} \left( \frac{(\pi(y) + \varepsilon_{y})^{2}}{\pi(y) - \varepsilon{y}} 
			       -  \pi(y) \right)  = \SmallO{\HH{X_0}}~. \nonumber
\end{align}
Combining the estimations \eqref{eq:lem5p2} and \eqref{eq:lem5p3} we obtain
$$
  \HH{X_0|X_{\tmix{\varepsilon}}} \geq \HH{X_0}(1 - \SmallO{1}) - \SmallO{\HH{X_0}}~,
$$
what results in
$$
	\frac{\HH{X_0|X_T}}{\HH{X_0}} \geq \frac{\HH{X_0}(1 - \SmallO{1}) - \SmallO{\HH{X_0}}}{\HH{X_0}} = 1 - \SmallO{1}~,
$$
as required.

We have set $q = \frac{f(n)}{\tmix{\varepsilon}}$ for arbitrary fixed $f(n) = \SmallO{1}$ and
$\varepsilon = \SmallO{\min\{\frac{1}{n^2}, \HH{X_0}\}}$. From Fact \ref{fact:MixingTime}
and Fact \ref{fact:RWMixing} we have $\tmix{\varepsilon} \leq n^{3} \log \varepsilon^{-1}$. Hence, there exists some
$g(n) = \SmallOmega{\max\{n^2, \frac{1}{\HH{X_0}}\}}$ dependent on $\varepsilon$ such that
$\tmix{\varepsilon} \leq n^{3} \log g(n)$ and
$$
  q = \frac{1}{h(n)} \cdot \frac{1}{n^{3} \log g(n)}~,
$$
where $h(n) = 1/f(n) = \SmallOmega{1}$.
\end{proof}

\begin{remark}
	As previously mentioned, the running time $T$ of the considered hiding algorithm follows geometric distribution with
	parameter $q$, hence the expected running time is $\E{T} = 1/q = h(n) \cdot n^{3} \log g(n)$, where $h(n)$ and $g(n)$
	are as in the proof of Theorem \ref{thm:unknowntopoMemorylessRnd}. If $H(X_0) = \BigOmega{\frac{1}{n^2}}$, as in the
	case of most distribution considered in practice, we can simply select $f(n)$ to be some function decreasing to 0
	arbitrary slowly and $\varepsilon$ such that $g(n) = c n^{3} \log{n}$ for some constant $c > 0$. In such cases
	the entropy of the distribution of agent's initial position has no impact on the upper bound on the algorithm's
	running time.
\end{remark}

It is clear that this algorithm works both for the scenario with single and many agents in the network
(in the latter the agents execute the algorithm independently using their own randomness).
The interesting question is whether it is possible to hide the initial state in the multi-agents 
case faster by taking advantage of performing simultaneously many random walks. As the speedup of multiple random walks
(in terms of cover time) in any graph remains a conjecture~\cite{AlonAKKLT11} we leave the speedup in terms of hiding
as a (possibly challenging) open question.

We conclude this section with a simple observation that the agent must have access to either memory, source of
randomness or the topology of the network in order to hide.
\begin{theorem}
	\label{thm:detNoMemo}
	In the model with unknown topology and with no memory there exists no well-hiding deterministic algorithm.
\end{theorem}	
\begin{proof}
	Take any hiding algorithm. If this algorithm never makes any move it obviously is not well-hiding. Otherwise observe
	that in the model without memory the move is decided only based on local observations (degree of the node) and some
	global information (value of $n$), hence every time the agent visits a node it will make the same decision. Assume
	that the agent decides to move from a node of degree $d$ via port $p$. We construct a graph from two stars with degree
	$d$ joined by an edge $e$ with port $p$ on both endpoints. Since the agent has no memory and no randomness it will end
	up in an infinite loop traversing edge $e$. Hence this algorithm cannot be regarded as well-hiding since it never
	terminates.
\end{proof}

\subsection{Unlimited Memory}
In this section we assume that the agent is endowed with unlimited memory that remains intact when the agent traverses
an edge. We first observe that a standard search algorithm (e.g., DFS) can be carried out in such a model. 

\begin{theorem}
	\label{thm:DFS}
	There exists a perfectly-hiding algorithm in the model with unlimited memory that needs $O(m)$ steps in any graph.
\end{theorem}
\begin{proof}
	The algorithm works as follows: it performs a DFS search of the graph (it is possible since the agent has memory) and
	then moves to the node with minimum ID.
\end{proof} 
Now we would like to show that $\Omega(m)$ steps are necessary for any well-hiding algorithm in this model. We will
construct a particular family of graphs such that for any well-hiding algorithm and any $n$ and $m$ we can find a graph
with $n$ nodes and $m$ edges in this family such that this algorithm will need on average $\Omega(m)$ steps.
\begin{theorem}
	\label{thm:unknownMemory}
	For a single agent and unknown network topology, for any $n$ and $m$ and any well hiding algorithm $\AAA$ there exists
	a port labeled graph $G$ with $n$ vertices and $m$ edges representing the underlying network and a distribution $\LLL$
	of agent's initial position on which the agent needs to perform the expected number of
	$\Omega(m)$ steps, where $m$ is the number of edges of $G$.
\end{theorem} 
\begin{proof}
	If $m = O(n)$ we can construct a graph in which $D = \Omega(m)$ and use Theorem~\ref{thm:known}. Now assume that
	$m = \omega(n)$ and consider a graph constructed by connecting a chain of $y$ cliques of size $x$. If $m=\omega(n)$
	we can find such $x,y$ that $x = \Theta(m/n)$ and $y = \Theta(n^2/m)$.

The adjacent cliques are connected by adding an additional vertex on two edges (one from each clique) and connecting
these vertices by an additional edge. We call this edge by \emph{bridge} and the vertices adjacent to a bridge will be
called \emph{bridgeheads}. Call by $\mathcal{G}_{x,y}$ the family of all such chains of cliques on $n$ nodes and
$m$ vertices (note that we take only such chains in which an edge contains at most one bridgehead). We want to calculate
the expected time of algorithm $\AAA$ to reach to the middle of the chain (if $y$ is even then the middle is the middle
bridge, otherwise it is the middle clique) on a graph chosen uniformly at random from family $\mathcal{G}_{x,y}$.
	
		
Clearly when the agent is traversing edges of the clique, each edge can contain a bridgehead with probability
$\frac{1}{\binom{x}{2}-1}$ hence with probability at least $1/2$ the agent needs to traverse $\frac{\binom{x}{2}-1}{2}$
different edges. Bridgeheads are chosen independently in each clique hence we can choose the constants so that by the
Chernoff Bound we have that with probability at least $\frac34$ the time to reach the middle of the chain is
$\BigOmega{y \cdot x^2} = \BigOmega{m}$ if $G$ is chosen uniformly at random from $\mathcal{G}_{x,y}$. By symmetry, this
holds for both endpoints (reaching middle from the first clique or the $y$-th one). Hence there exists
$G^* \in \mathcal{G}_{x,y}$ such that the time of $\AAA$ to reach the middle from both endpoints on $G^*$ is at least
$c \cdot m$, for some constant $c>0$ with probability at least $\frac34$. 
Now by Lemma~\ref{lem:known} such algorithm is not well-hiding on graph $G^*$ if the number of steps is at most
$c\cdot m$.
		
\end{proof}


%
%
%
%
%

\section{Conclusions and Further Research}\label{sect:conclusion}
We introduced and motivated the problem of location hiding. We also discussed efficient algorithms and lower bounds for
some model settings. Nevertheless we believe that some interesting questions are left unanswered. 

First, we plan to deeper understand the relation of location hiding problem with classic, fundamental problems
like rendez-vous or patrolling. 

Note that we assumed the energy complexity is the maximal energetic expenditure over all agents. In this case 
both complexities (time and energy) are equivalent except some artificial examples. 
In some cases however it would be more adequate to consider a total energy used by all stations. 
This leads to substantially different analysis that at least in some cases seems to use more elaborated techniques. 

Another line of research is the model with dynamic topology that may change during the execution of the protocol.
Similarly, we believe that it would be interesting to investigate the model with the weaker adversary that is given only 
partial knowledge of the graph topology and the actual assignment of agents. On the other hand, one may consider 
more powerful adversaries being able to observe some chosen part of the network for a given period of time.
It would be also worth considering how high level of security can be achieved if each agent is able to
perform only $\BigO{1}$ steps.

Motivated by the fact that the mobile devices are actually similar objects, we considered the setting where they are
indistinguishable. It would be useful to study the case when the adversary
can distinguish between different agents.

Finally, we also suppose that it would be also interesting for applications to construct more efficient protocols for
given classes of graphs with some common characteristic (e.g., lines, trees) and algorithms desired for restricted
distributions of $X_0$. For example it is clear that if we know that the initial assignment of agents is uniform
we can design more efficient (in expectation) and realistic algorithms.

\appendix
\section{Appendix - Information Theory}\label{sect:apend}

Let us recall some basic definitions and facts from Information Theory that can be found e.g. in \cite{Thomas,Lukasz}. 
Let us mention that in all cases below $\log$ will always denote the base-2 logarithm and as a consequence all the measures are 
expressed in \textit{bits}. 

\begin{definition}[Entropy]
	For  a discrete random variable $X \colon \XXX \rightarrow \RR$ we define entropy as
	$$
		\HH{X} = -\sum_{x \in \XXX} \Pr[X=x] \log \Pr[X=x]~.
	$$ 
\end{definition}

\begin{definition}[Conditional Entropy]
	If $X$ and $Y$ are two discrete random variables defined on the spaces $\XXX$ and $\YYY$, respectively,
	the conditional entropy is defined as
	$$
		\HH{X|Y} = -\sum_{y \in \YYY} \Pr[Y=y] \sum_{x \in \XXX} \Pr[X=x|Y=y] \log \Pr[X=x|Y=y]~.
	$$
\end{definition}

The following simple fact states that conditioning on the other random variable cannot increase the entropy.
\begin{fact}
	\label{fact:condReduce}
	For any random variables $X$ and $Y$
	$$
		\HH{X|Y} \leq \HH{X}
	$$
	and the equality holds if and only if $X$ and $Y$ are independent.
\end{fact}

\begin{definition}[Relative Entropy]
	\label{def:relEnt}
	Let $X$ and $Y$ be two discrete random variables defined on the common probability space $\XXX$ with probability
	mass functions $p(x)$ and $q(x)$, respectively. We define the relative entropy (also known as Kullback-Leibler distance)
	between $p(x)$ and $q(x)$ as
	$$
		\DD{p}{q} = \sum_{x \in \XXX} p(x) \log \frac{p(x)}{q(x)}
	$$
	with the convention that $0 \log \frac{0}{0} = 0$, $0 \log \frac{0}{q} = 0$ and $0 \log \frac{p}{0} = \infty$.
\end{definition}

\begin{fact}[Information inequality]
	\label{fact:RelEnt1}
  Let $p(x)$ and $q(x)$ be probability mass functions of two discrete random variables
	$X, Y \colon \XXX \rightarrow \RR$. Then
	$$
		\DD{p}{q} \geq 0
	$$
	with equality if and only if $\forall x \in \XXX$ $p(x) = q(x)$.
\end{fact}

The following fact gives an upper bound on relative entropy $\DD{p}{q}$. The proof can be found in \cite{DRAGOMIRrelEnt}.
\begin{fact}[Theorem 1 in \cite{DRAGOMIRrelEnt}]
	\label{thm:RelEnt2}
	Let $p(x)$, $q(x) > 0$ be probability mass functions of two discrete random variables $X$ and $Y$, respectively,
	defined on the space $\XXX$. Then
	$$
	  \DD{p}{q} \leq \frac{1}{\ln 2} \left( \sum_{x \in \XXX} \frac{p^{2}(x)}{q(x)} - 1 \right)~.
	$$

\end{fact}

\begin{definition}[Mutual Information]\label{eq:MutInf1}
If $X$ and $Y$ are two discrete random variables defined on the spaces $\XXX$ and $\YYY$, respectively,
then the mutual information of $X$ and $Y$ is defined as
\begin{equation}
	  \MutInf{X}{Y} = \sum_{x \in \XXX} \sum_{y \in \YYY} \Pr[X=x, Y=y] \log\left(\frac{\Pr[X=x, Y=y]}{\Pr[X=x] \Pr[Y=y]}\right)~.
\end{equation}
\end{definition}

\begin{fact}
\label{eq:MutInf2}
 For any discrete random variables $X,Y$ 
\begin{itemize}
	\item $0 \leq \MutInf{X}{Y} \leq \min\{\HH{X}, \HH{Y}\}$ and the first equality holds 
				if and only if  random variables $X$ and $Y$ are independent,
	\item $\MutInf{X}{Y} = \MutInf{Y}{X} = \HH{X} - \HH{X|Y} = \HH{Y} - \HH{Y|X}$.
	\item $\MutInf{X}{Y} = \DD{p(x,y)}{p(x)p(y)}$ where $p(x,y)$ denotes the joint distribution, and $p(x)p(y)$
	      the product distribution of $X$ and $Y$.
\end{itemize}
\end{fact}

\section{Appendix - Markov Chains}
\label{sect:apendMC}

In this appendix we recall some fundamental definitions and facts from the theory of Markov processes and random walks,
including some basic properties
of Markov chains related to information theory. They can be found e.g. in \cite{LevinPeresWilmerMCMT,lovasz1993random,Thomas}. 

Unless otherwise stated, we will consider only time-homogeneous chains, i.e. such chains for which the transition
probabilities between states do not change with time and therefore can be described by some stochastic matrix $P$.

\begin{definition}[Total variation distance]
	For two probability distributions $\mu$ and $\nu$ on the space $\XXX$ we define the total variation distance
	between $\mu$ and $\nu$ as
	$$
		\dTV{\mu}{\nu} = \max_{A \subseteq \XXX} |\mu(A) - \nu(A)|~.
	$$
\end{definition}

\begin{fact}
  Let $\mu$ and $\nu$ be probability distribution on $\XXX$. Then
	$$
		\dTV{\mu}{\nu} = \frac{1}{2} \sum_{x \in \XXX} |\mu(x) - \nu(x)|~.
	$$
\end{fact}

\begin{definition}
Let $P^{t}(x_0, \cdot)$ denote the distribution of an ergodic Markov chain $M$ on finite space $\XXX$ in step $t$
when starting in the state $x_0$. Let $\pi$ be the stationary distribution of $M$. We define 
$$
  d(t) = \max_{x \in \XXX} \dTV{P^{t}(x, \cdot)}{\pi}
$$
and
$$
	\bar{d}(t) = max_{x, y \in \XXX} \dTV{P^{t}(x, \cdot)}{P^{t}(y, \cdot)}~.
$$
\end{definition}

\begin{fact}
	\label{fact:dMC1}
	Denoting by $\PPP$ the family of all probability distributions on $\XXX$, the following relations for $d(t)$
	and $\bar{d}(t)$ hold.
	\begin{itemize}
		\item $d(t) \leq \bar{d}(t) \leq 2 d(t)$,
		\item $d(t) = \sup_{\mu \in \PPP} \dTV{\mu P^{t}}{\pi}$.
		\item $d(t) = \sup_{\mu, \nu \in \PPP} \dTV{\mu P^{t}}{\nu P^{t}}$.
	\end{itemize}
\end{fact}

\begin{definition}[Mixing time]
	\label{def:MixingTime}
	For an ergodic Markov chain $M$ on finite space $\XXX$ we define  the mixing time as
	$$
		\tmix{\varepsilon} = \min \{t \colon d(t) \leq \varepsilon \}
	$$
	and
	$$
	  \mix = \tmix{1/4}~.
	$$
\end{definition}

\begin{fact}
	\label{fact:MixingTime}
	For any $\varepsilon > 0$
	$$
		\tmix{\varepsilon} \leq \left\lfloor \log \varepsilon^{-1} \right\rfloor \mix~.
	$$
\end{fact}

\begin{definition}[Random walk on graph G]
	\label{def:rw}
	The random walk on a connected graph $G= (V,E)$ with $n$ nodes and $m$ edges is a Markov chain on the set of vertices
	with transition probabilities defined by
	$$
		p_{ij} = \Pr[X_{t+1} = v_j | X_t = v_i] =
		\begin{cases}
				1/\degV{v_i}, & \text{if } {v_i,v_j} \in E~,\\
				0,           & \text{otherwise,}
		\end{cases}
	$$
	where $\degV{v_i}$ denotes the degree of vertex $v_i$.
	
	By the lazy random walk we mean the random walk which, in every time $t$, with probability $1/2$ remains in current
	vertex or performs one step of a simple random walk, i.e. tha chain with transition probabilities
	$$
		p_{ij} = \Pr[X_{t+1} = v_j | X_t = v_i] =
		\begin{cases}
				1/2, & \text{if } {v_j = v_i}~,\\
				1/(2\ \degV{v_i}), & \text{if } {v_i,v_j} \in E~,\\
				0,           & \text{otherwise.}
		\end{cases}
	$$
\end{definition}
It is well known (see e.g. \cite{LevinPeresWilmerMCMT,lovasz1993random}) that if $G$ is not bipartite, the simple
random walk is reversible and has unique stationary distribution of the form $\pi(v_i) = \degV{v_i}/2m$. 
The same is true for lazy random walk and an arbitrary connected $G$.

The following fact gives an upper bound on the mixing time for random walks. It follows e.g. from the theorem 10.14 in
\cite{LevinPeresWilmerMCMT} and the fact the cover time upper bounds the maximum hitting time and the fact that the maximum cover time for undirected graphs is $O(n^3)$~\cite{FeigeUpper}.
\begin{fact}
	\label{fact:RWMixing}
  For a lazy random walk on an arbitrary connected graph $G$ with $n$ vertices 
	$\mix = \BigO{n^3}.$
\end{fact}

In the following, we will recall some facts on Markov chains from information theory. For references, see e.g.
\cite{Thomas,Cover1994SecondLaw,Renyi1961}.

\begin{fact}
	\label{fact:MCInfTh}
	Let $M = (X_0, X_1, \ldots)$ be an ergodic Markov chain on finite space $\XXX$ with transition matrix $P$
	and stationary distribution $\pi$.
	\begin{itemize}
		\item For any two probability distributions $\mu$ and $\nu$ on $\XXX$ the relative entropy $\DD{\mu P^{t}}{\nu P^{t}}$
		      decreases with $t$, i.e. $\DD{\mu P^{t}}{\nu P^{t}} \geq \DD{\mu P^{t+1}}{\nu P^{t+1}}$.
		\item For any initial distribution $\mu$ the relative entropy $\DD{\mu P^{t}}{\pi}$ between the distribution of
		      $M$ at time $t$ and the stationary distribution decreases with $t$. Furthermore, 
					$\lim_{t \rightarrow \infty} \DD{\mu P^{t}}{\pi} = 0$.
		\item The conditional entropy $\HH{X_0|X_t}$ is increasing in $t$.
	\end{itemize}
\end{fact}

%
%
%
 \bibliographystyle{./splncs-url}
\bibliography{./bibliography} 
 
\end{document}